\documentclass[envcountsame,runningheads,a4paper]{llncs}
\usepackage{pscproc2}
\usepackage{latexsym}
\usepackage{amssymb}
\usepackage{amsmath}
\usepackage{ifthen}
\usepackage{stmaryrd}
\usepackage{booktabs}
\usepackage{enumerate}
\usepackage{tikz}
\usepackage{subfigure}
\newcommand{\oneset}[1]{\left\{\mathinner{#1}\right\}}
\newcommand{\smallset}[1]{\left\{\mathinner{#1}\right\}}

\newcommand{\abs}[1]{\left|\mathinner{#1}\right|}

\newcommand{\N}{\mathbb{N}}

\newcommand{\coloneq}{\mathrel{\mathop:}=}

\newcommand{\context}{\mathrm{context}}
\newcommand{\last}{\mathrm{last}}

\newcommand{\first}{\mathrm{first}}
\newcommand{\rightshift}{r}

\newcommand{\lbl}{\lambda}
\newcommand{\BWT}{\mathrm{BWT}}
\newcommand{\ST}{\mathrm{ST}}
\newcommand{\bijBWT}{\mathrm{BWTS}}
\newcommand{\bijST}{\mathrm{LST}}
\newcommand{\bijM}{\mathrm{LM}}

\newtheorem{fact}[theorem]{Fact}

\begin{document}

\title{On Bijective Variants of the \\ Burrows-Wheeler Transform}

\author{Manfred Kuf{}leitner
\institute{Universit{\"a}t Stuttgart, FMI, \\ Universit{\"a}tsstr.\ 38, 70569 Stuttgart, Germany \\
\email{kufleitner@fmi.uni-stuttgart.de}}}

%\date{\today} % optional

\maketitle

\begin{abstract}
  The sort transform (ST) is a modification of the Burrows-Wheeler
  transform (BWT). Both transformations map an arbitrary word of
  length $n$ to a pair consisting of a word of length $n$ and an index
  between $1$ and $n$. The BWT sorts all rotation conjugates of the
  input word, whereas the ST of order $k$ only uses the first $k$
  letters for sorting all such conjugates. If two conjugates start with the
  same prefix of length $k$, then the indices of the rotations are
  used for tie-breaking.  Both transforms output the sequence of the
  last letters of the sorted list and the index of the input within
  the sorted list.  In this paper, we discuss a bijective variant of
  the BWT (due to Scott), proving its correctness and relations to
  other results due to Gessel and Reutenauer (1993) and Crochemore,
  D{\'e}sarm{\'e}nien, and Perrin (2005).  Further, we present a novel
  bijective variant of the ST. 
%  In saying bijective, we mean that the
%   transforms map any word of length $n$ to another word of length $n$,
%   and that further, the transformed word uniquely determines the
%   original.
\end{abstract}

\section{Introduction}
\label{sec:intro}

The Burrows-Wheeler transform (BWT) is a widely used preprocessing
technique in lossless data compression \cite{bw94tr}. It brings every
word into a form which is likely to be easier to compress
\cite{man01acm}. Its compression performance is almost as good as PPM
(prediction by partial matching) schemes \cite{cw84toc} while its
speed is comparable to that of Lempel-Ziv algorithms
\cite{lz77,lz78}. Therefore, BWT based compression schemes are a very
reasonable trade-off between running time and compression ratio.

In the classic setting, the BWT maps a word of length $n$ to a word of
length $n$ and an index (comprising $O(\log n)$ bits). Thus,
the BWT is not bijective and hence, it is introducing new redundancies
to the data, which is cumbersome and undesired in applications of data
compression or cryptography. Instead of using an index, a very common
technique is to assume that the input has a unique end-of-string
symbol \cite{bk00tc,man01acm}.  Even though this often simplifies
proofs or allows speeding up the algorithms, the use of an
end-of-string symbol introduces new redundancies (again $O(\log n)$
bits are required for coding the end-of-string symbol).

We discuss bijective versions of the BWT which are one-to-one
correspondences between words of length $n$. In particular, no index
and no end-of-string symbol is needed. Not only does bijectivity save
a few bits, for example, it also increases data security when
cryptographic procedures are involved; it is more natural and it can
help us to understand the BWT even better.  Moreover, the bijective
variants give us new possibilities for enhancements; for example, in
the bijective BWT different orders on the letters can be used for the
two main stages.

Several variants of the BWT have been introduced
\cite{aa04ijcm,mrrs07tcs}. An overview can be found in the textbook by
Adjeroh, Bell, and Mukherjee \cite{abm08book}. One particularly
important variant for this paper is the sort transform (ST), which is
also known under the name Schindler transform
\cite{schindler97dcc}. In the original paper, the inverse of the ST is
described only very briefly. More precise descriptions and improved
algorithms for the inverse of the ST have been proposed recently
\cite{nz06dcc,nz07ieee,nzc08cpm}.  As for the BWT, the ST also
involves an index or an end-of-string symbol. In particular, the ST is
not onto and it introduces new redundancies.

The bijective BWT was discovered and first described by Scott (2007),
but his exposition of the algorithm was somewhat cryptic, and was not
appreciated as such. In particular, the fact that this transform is
based on the Lyndon factorization went unnoticed by Scott.  Gil and
Scott \cite{gs09submitted} provided an accessible description of the
algorithm.  Here, we give an alternative description, a proof of its
correctness, and more importantly, draw connections between Scott's
algorithm and other results in combinatorics on words. Further, this
variation of the BWT is used to introduce techniques which are
employed at the bijective sort transform, which makes the main
contribution of this paper.  The forward transform of the bijective ST
is rather easy, but we have to be very careful with some
details. Compared with the inverse of the bijective BWT, the inverse
of the bijective ST is more involved.

\bigskip

\noindent
\textbf{Outline.} \ 
The paper is organized as follows. In
Section~\ref{sec:prelim} we fix some notation and repeat  basic
facts about combinatorics on words.  On our way to the bijective sort
transform (Section~\ref{sec:lst}) we investigate the BWT
(Section~\ref{sec:bwt}), the bijective BWT (Section~\ref{sec:lbwt}),
and the sort transform (Section~\ref{sec:st}).  We give full
constructive proofs for the injectivity of the respective transforms.
Each section ends with a running example which illustrates the
respective concepts.  Apart from basic combinatorics on words, the
paper is completely self-contained.

\section{Preliminaries}
\label{sec:prelim}

Throughout this paper we fix the finite non-empty alphabet $\Sigma$
and assume that $\Sigma$ is equipped with a linear order $\leq$.  A
\emph{word} is a sequence $a_1 \cdots a_n$ of letters $a_i \in
\Sigma$, $1 \leq i \leq n$.  The set of all such sequences is denoted
by $\Sigma^*$; it is the free monoid over~$\Sigma$ with concatenation
as composition and with the empty word $\varepsilon$ as neutral
element. The set $\Sigma^+ = \Sigma^* \setminus
\smallset{\varepsilon}$ consists of all non-empty words.  For words
$u,v$ we write $u \leq v$ if $u=v$ or if $u$ is lexicographically
smaller than $v$ with respect to the order $\leq$ on the letters.  Let
$w = a_1 \cdots a_n \in \Sigma^+$ be a non-empty word with letters
$a_i \in \Sigma$. The \emph{length} of $w$, denoted by $\abs{w}$, is
$n$.  The empty word is the unique word of length $0$.  We can think
of $w$ as a labeled linear order: position $i$ of $w$ is labeled by
$a_i \in \Sigma$ and in this case we write $\lbl_w(i) = a_i$, so each
word $w$ induces a labeling function $\lbl_w$.
%
% Frequently, we identify functions $f : X \to Y$ with there extensions
% to sequences $f^* : X^* \to Y^*$ where $f^*(x_1 \cdots x_n) = f(x_1)
% \cdots f(x_n)$ for $x_1, \ldots, x_n \in X$. For example, we write
% $\lbl_w(2,4,6)$ instead of $\lbl_w(2) \lbl_w(4) \lbl_w(6)$.
%
The first letter $a_1$ of $w$ is denoted by $\first(w)$ while the last
letter $a_n$ is denoted by $\last(w)$.  The \emph{reversal} of a word
$w$ is $\overline{w} = a_n \cdots a_1$.  We say that two words $u,v$
are \emph{conjugate} if $u = st$ and $v = ts$ for some words $s,t$,
i.e., $u$ and $v$ are cyclic shifts of one another. The $j$-fold
concatenation of $w$ with itself is denoted by $w^j$.  A word $u$ is a
\emph{root} of $w$ if $w = u^j$ for some $j \in \N$.  A word $w$ is
\emph{primitive} if $w = u^j$ implies $j = 1$ and hence $u = w$, i.e.,
$w$ has only the trivial root $w$.

The \emph{right-shift} of $w = a_1 \cdots a_n$ is $\rightshift(w) =
a_n a_1 \cdots a_{n-1}$ and the $i$-fold right
shift~$\rightshift^i(w)$ is defined inductively by $\rightshift^0(w) =
w$ and $\rightshift^{i+1}(w) = \rightshift(\rightshift^{i}(w))$. We
have $\rightshift^{i}(w) = a_{n-i+1} \cdots a_n a_1 \cdots a_{n-i}$
for $0 \leq i < n$. The word $\rightshift^i(w)$ is also well-defined
for~$i \geq n$ and then $\rightshift^i(w) = \rightshift^j(w)$ where $j
= i \bmod n$. We define the \emph{ordered conjugacy class} of a word
$w \in \Sigma^n$ as $[w] = (w_1, \ldots, w_n)$ where $w_i =
\rightshift^{i-1}(w)$. It is convenient to think of~$[w]$ as a cycle
of length $n$ with a pointer to a distinguished starting
position. Every position $i$, $1 \leq i \leq n$, on this cycle is
labeled by $a_i$. In particular, $a_1$ is a successor of $a_n$ on this
cycle since the position $1$ is a successor of the position $n$. The
mapping $\rightshift$ moves the pointer to its predecessor. The
(unordered) conjugacy class of $w$ is the multiset $\oneset{w_1,
  \ldots, w_n}$. Whenever there is no confusion, then by abuse of
notation we also write $[w]$ to denote the (unordered) conjugacy class
of $w$. For instance, this is the case if $w$ is in some way
distinguished within its conjugacy class, which is true if $w$ is a
Lyndon word.
A \emph{Lyndon word} is a non-empty word which is the unique
lexicographic minimal element within its conjugacy class. More
formally, let $[w] = (w, w_2, \ldots, w_n)$, then $w\in \Sigma^+$ is a
Lyndon word if $w < w_i$ for all $i \in \oneset{2, \ldots, n}$. Lyndon
words have a lot of nice properties \cite{Lot83}. For instance, Lyndon
words are primitive.  Another interesting fact is the following.

\begin{fact}[Chen, Fox, and Lyndon \cite{cfl58ann}]\label{thm:lyndon}
  Every word $w \in \Sigma^+$ has a unique factorization
%  \begin{equation*}
$
    w = v_s \cdots v_1
$
%  \end{equation*}
  such that $v_1 \leq \cdots \leq v_s$ is a non-decreasing sequence of
  Lyndon words.
\end{fact}

An alternative formulation of the above fact is that every word $w$
has a unique factorization $w = v_s^{n_s} \cdots v_1^{n_1}$ where $n_i
\geq 1$ for all $i$ and where $v_1 < \cdots < v_s$ is a strictly
increasing sequence of Lyndon words.  The factorization of $w$ as in
Fact~\ref{thm:lyndon} is called the \emph{Lyndon factorization} of
$w$. It can be computed in linear time using Duval's algorithm
\cite{duv83ja}.

Suppose we are given a multiset $V = \oneset{v_1, \ldots, v_s}$ of
Lyndon words enumerated in non-decreasing order $v_1 \leq \cdots \leq
v_s$.  Now, $V$ uniquely determines the word $w = v_s \cdots
v_1$. Therefore, the Lyndon factorization induces a one-to-one
correspondence between arbitrary words of length $n$ and multisets of
Lyndon words of total length $n$. Of course, by definition of Lyndon
words, the multiset $\oneset{v_1, \ldots, v_s}$ of Lyndon words and
the multiset $\oneset{[v_1], \ldots, [v_s]}$ of conjugacy classes of
Lyndon words are also in one-to-one correspondence.

We extend the order $\leq$ on $\Sigma$ as follows to non-empty words.
Let $w^{\omega} = w w w \cdots$ be the infinite sequences obtained as
the infinite power of $w$. For $u,v \in \Sigma^+$ we write $u
\leq^\omega v$ if either $u^{\omega} = v^{\omega}$ or $u^{\omega} = paq$
and $v^{\omega} = pbr$ for $p \in \Sigma^*$, $a,b \in \Sigma$ with
$a<b$, and infinite sequences $q,r$; phrased differently, $u
\leq^\omega v$ means that the infinite sequences~$u^\omega$ and
$v^\omega$ satisfy $u^\omega \leq v^\omega$.  If $u$ and $v$ have the
same length, then $\leq^\omega$ coincides with the lexicographic order
induced by the order on the letters. For arbitrary words,
$\leq^\omega$ is only a preorder since for example $u \leq^\omega uu$
and $uu \leq^\omega u$.  On the other hand, if $u \leq^\omega v$ and
$v \leq^\omega u$ then $u^{\abs{v}} = v^{\abs{u}}$. Hence, by the
periodicity lemma \cite{FW65}, there exists a common root $p \in
\Sigma^+$ and $g,h \in \N$ such that $u = p^g$ and $v = p^h$.  Also
note that $b \leq ba$ whereas $ba \leq^{\omega} b$ for $a<b$.

Intuitively, the \emph{context of order $k$} of $w$ is the sequence of
the first $k$ letters of $w$. We want this notion to be well-defined
even if $\abs{w} < k$. To this end let $\context_k(w)$ be the prefix
of length $k$ of $w^\omega$, i.e., $\context_k(w)$ consists of the
first $k$ letters on the cycle $[w]$.  Note that our definition of a
context of order $k$ is left-right symmetric to the corresponding
notion used in data compression.  This is due to the fact that typical
compression schemes are applying the BWT or the ST to the reversal of
the input.

An important construction in this paper is the \emph{standard
  permutation} $\pi_w$ on the set of positions $\oneset{1,\ldots, n}$
induced by a word $w = a_1 \cdots a_n \in \Sigma^n$
\cite{gr93jct}. The first step is to introduce a new order $\preceq$
on the positions of $w$ by sorting the letters within $w$ such that
identical letters preserve their order.  More formally, the linear
order $\preceq$ on $\oneset{1,\ldots,n}$ is defined as follows: $i
\preceq j$ if
\begin{equation*}
  a_i < a_j \qquad \text{or} \qquad a_i = a_j \,\text{ and }\, i \leq j.
\end{equation*}
Let $j_1 \prec \cdots \prec j_n$ be the linearization of $\oneset{1,
  \ldots, n}$ according to this new order. Now, the standard
permutation $\pi_w$ is defined by $\pi_w(i) = j_i$.

\begin{example}\label{exa:start}
  Consider the word $w = bcbccbcbcabbaaba$ over the ordered alphabet
  $a < b < c$. We have $\abs{w} = 16$. Therefore, the positions in $w$
  are $\oneset{1, \ldots, 16}$.  For instance, the label of position
  $6$ is $\lbl_w(6) = b$. Its Lyndon factorization is $w = bcbcc \cdot
  bc \cdot bc \cdot abb \cdot aab \cdot a$.  The context of order $7$
  of the prefix $bcbcc$ of length $5$ is $bcbccbc$ and the context of
  order $7$ of the factor $bc$ is $bcbcbcb$. For computing the
  standard permutation we write $w$ column-wise, add positions, and
  then sort the pairs lexicographically:
  \begin{center} \footnotesize
  \begin{tabular}{|c|c|c|}
    \hline
    \ word $w$ \ & \ $w$ with positions \ & \ sorted$\ $  \\
    \hline
    $b$ & $(b,1)$  & $(a,10)$ \\[-1.2pt]
    $c$ & $(c,2)$  & $(a,13)$ \\[-1.2pt]
    $b$ & $(b,3)$  & $(a,14)$ \\[-1.2pt]
    $c$ & $(c,4)$  & $(a,16)$ \\[-1.2pt]
    $c$ & $(c,5)$  & $(b,1)$  \\[-1.2pt]
    $b$ & $(b,6)$  & $(b,3)$  \\[-1.2pt]
    $c$ & $(c,7)$  & $(b,6)$  \\[-1.2pt]
    $b$ & $(b,8)$  & $(b,8)$  \\[-1.2pt]
    $c$ & $(c,9)$  & $(b,11)$ \\[-1.2pt]
    $a$ & $(a,10)$ & $(b,12)$ \\[-1.2pt]
    $b$ & $(b,11)$ & $(b,15)$ \\[-1.2pt]
    $b$ & $(b,12)$ & $(c,2)$  \\[-1.2pt]
    $a$ & $(a,13)$ & $(c,4)$  \\[-1.2pt]
    $a$ & $(a,14)$ & $(c,5)$  \\[-1.2pt]
    $b$ & $(b,15)$ & $(c,7)$  \\[-1.2pt]
    $a$ & $(a,16)$ & $(c,9)$  \\
    \hline
  \end{tabular}
  \end{center}
  This yields the standard permutation
  \begin{equation*}
    \pi_w = \left( 
      \begin{array}{*{16}{p{5.5mm}}}
        1 & 2 & 3 & 4 & 5 & 6 & 7 & 8 & 9 & 10 & 11 & 12 & 13 & 14 & 15 & 16 \\
        10 & 13 & 14 & 16 & 1 & 3 & 6 & 8 & 11 & 12 & 15 & 2 & 4 & 5 & 7 & 9
      \end{array}
      \right).
  \end{equation*}
  The conjugacy class $[w]$ of $w$ is depicted in
  Figure~\ref{sfg:w}; the $i$-th word in $[w]$ is written in the
  $i$-th row.  The last column of the matrix for $[w]$ is the reversal
  $\overline{w}$ of $w$.
\end{example}

\begin{figure}[t]
  \centering \footnotesize
  \subfigure[{Conjugacy class $[w]$}]{\footnotesize
  \begin{tabular}{r|*{16}{p{1.75mm}}|}
    \cline{2-17}
    \textsf{1}  & $b$ & $c$ & $b$ & $c$ & $c$ & $b$ & $c$ & $b$ & $c$ & $a$ & $b$ & $b$ & $a$ & $a$ & $b$ & $a$ \\[-1mm]
    \textsf{2}  & $a$ & $b$ & $c$ & $b$ & $c$ & $c$ & $b$ & $c$ & $b$ & $c$ & $a$ & $b$ & $b$ & $a$ & $a$ & $b$ \\[-1mm]
    \textsf{3}  & $b$ & $a$ & $b$ & $c$ & $b$ & $c$ & $c$ & $b$ & $c$ & $b$ & $c$ & $a$ & $b$ & $b$ & $a$ & $a$ \\[-1mm]
    \textsf{4}  & $a$ & $b$ & $a$ & $b$ & $c$ & $b$ & $c$ & $c$ & $b$ & $c$ & $b$ & $c$ & $a$ & $b$ & $b$ & $a$ \\[-1mm]
    \textsf{5}  & $a$ & $a$ & $b$ & $a$ & $b$ & $c$ & $b$ & $c$ & $c$ & $b$ & $c$ & $b$ & $c$ & $a$ & $b$ & $b$ \\[-1mm]
    \textsf{6}  & $b$ & $a$ & $a$ & $b$ & $a$ & $b$ & $c$ & $b$ & $c$ & $c$ & $b$ & $c$ & $b$ & $c$ & $a$ & $b$ \\[-1mm]
    \textsf{7}  & $b$ & $b$ & $a$ & $a$ & $b$ & $a$ & $b$ & $c$ & $b$ & $c$ & $c$ & $b$ & $c$ & $b$ & $c$ & $a$ \\[-1mm]
    \textsf{8}  & $a$ & $b$ & $b$ & $a$ & $a$ & $b$ & $a$ & $b$ & $c$ & $b$ & $c$ & $c$ & $b$ & $c$ & $b$ & $c$ \\[-1mm]
    \textsf{9}  & $c$ & $a$ & $b$ & $b$ & $a$ & $a$ & $b$ & $a$ & $b$ & $c$ & $b$ & $c$ & $c$ & $b$ & $c$ & $b$ \\[-1mm]
    \textsf{10} & $b$ & $c$ & $a$ & $b$ & $b$ & $a$ & $a$ & $b$ & $a$ & $b$ & $c$ & $b$ & $c$ & $c$ & $b$ & $c$ \\[-1mm]
    \textsf{11} & $c$ & $b$ & $c$ & $a$ & $b$ & $b$ & $a$ & $a$ & $b$ & $a$ & $b$ & $c$ & $b$ & $c$ & $c$ & $b$ \\[-1mm]
    \textsf{12} & $b$ & $c$ & $b$ & $c$ & $a$ & $b$ & $b$ & $a$ & $a$ & $b$ & $a$ & $b$ & $c$ & $b$ & $c$ & $c$ \\[-1mm]
    \textsf{13} & $c$ & $b$ & $c$ & $b$ & $c$ & $a$ & $b$ & $b$ & $a$ & $a$ & $b$ & $a$ & $b$ & $c$ & $b$ & $c$ \\[-1mm]
    \textsf{14} & $c$ & $c$ & $b$ & $c$ & $b$ & $c$ & $a$ & $b$ & $b$ & $a$ & $a$ & $b$ & $a$ & $b$ & $c$ & $b$ \\[-1mm]
    \textsf{15} & $b$ & $c$ & $c$ & $b$ & $c$ & $b$ & $c$ & $a$ & $b$ & $b$ & $a$ & $a$ & $b$ & $a$ & $b$ & $c$ \\[-1mm]
    \textsf{16} & $c$ & $b$ & $c$ & $c$ & $b$ & $c$ & $b$ & $c$ & $a$ & $b$ & $b$ & $a$ & $a$ & $b$ & $a$ & $b$ \\
    \cline{2-17}
  \end{tabular}
  \label{sfg:w}
  } \ 
  \subfigure[Lexicographically sorted]{\footnotesize
  \begin{tabular}{r|*{16}{p{1.75mm}}|}
    \cline{2-17}
    \textsf{5} & $a$ & $a$ & $b$ & $a$ & $b$ & $c$ & $b$ & $c$ & $c$ & $b$ & $c$ & $b$ & $c$ & $a$ & $b$ & $b$ \\[-1mm]
    \textsf{4} & $a$ & $b$ & $a$ & $b$ & $c$ & $b$ & $c$ & $c$ & $b$ & $c$ & $b$ & $c$ & $a$ & $b$ & $b$ & $a$ \\[-1mm]
    \textsf{8} & $a$ & $b$ & $b$ & $a$ & $a$ & $b$ & $a$ & $b$ & $c$ & $b$ & $c$ & $c$ & $b$ & $c$ & $b$ & $c$ \\[-1mm]
    \textsf{2} & $a$ & $b$ & $c$ & $b$ & $c$ & $c$ & $b$ & $c$ & $b$ & $c$ & $a$ & $b$ & $b$ & $a$ & $a$ & $b$ \\[-1mm]
    \textsf{6} & $b$ & $a$ & $a$ & $b$ & $a$ & $b$ & $c$ & $b$ & $c$ & $c$ & $b$ & $c$ & $b$ & $c$ & $a$ & $b$ \\[-1mm]
    \textsf{3} & $b$ & $a$ & $b$ & $c$ & $b$ & $c$ & $c$ & $b$ & $c$ & $b$ & $c$ & $a$ & $b$ & $b$ & $a$ & $a$ \\[-1mm]
    \textsf{7} & $b$ & $b$ & $a$ & $a$ & $b$ & $a$ & $b$ & $c$ & $b$ & $c$ & $c$ & $b$ & $c$ & $b$ & $c$ & $a$ \\[-1mm]
    \textsf{10} & $b$ & $c$ & $a$ & $b$ & $b$ & $a$ & $a$ & $b$ & $a$ & $b$ & $c$ & $b$ & $c$ & $c$ & $b$ & $c$ \\[-1mm]
    \textsf{12} & $b$ & $c$ & $b$ & $c$ & $a$ & $b$ & $b$ & $a$ & $a$ & $b$ & $a$ & $b$ & $c$ & $b$ & $c$ & $c$ \\[-1mm]
    \textsf{1} & $b$ & $c$ & $b$ & $c$ & $c$ & $b$ & $c$ & $b$ & $c$ & $a$ & $b$ & $b$ & $a$ & $a$ & $b$ & $a$ \\[-1mm]
    \textsf{15} & $b$ & $c$ & $c$ & $b$ & $c$ & $b$ & $c$ & $a$ & $b$ & $b$ & $a$ & $a$ & $b$ & $a$ & $b$ & $c$ \\[-1mm]
    \textsf{9} & $c$ & $a$ & $b$ & $b$ & $a$ & $a$ & $b$ & $a$ & $b$ & $c$ & $b$ & $c$ & $c$ & $b$ & $c$ & $b$ \\[-1mm]
    \textsf{11} & $c$ & $b$ & $c$ & $a$ & $b$ & $b$ & $a$ & $a$ & $b$ & $a$ & $b$ & $c$ & $b$ & $c$ & $c$ & $b$ \\[-1mm]
    \textsf{13} & $c$ & $b$ & $c$ & $b$ & $c$ & $a$ & $b$ & $b$ & $a$ & $a$ & $b$ & $a$ & $b$ & $c$ & $b$ & $c$ \\[-1mm]
    \textsf{16} & $c$ & $b$ & $c$ & $c$ & $b$ & $c$ & $b$ & $c$ & $a$ & $b$ & $b$ & $a$ & $a$ & $b$ & $a$ & $b$ \\[-1mm]
    \textsf{14} & $c$ & $c$ & $b$ & $c$ & $b$ & $c$ & $a$ & $b$ & $b$ & $a$ & $a$ & $b$ & $a$ & $b$ & $c$ & $b$ \\
    \cline{2-17}
  \end{tabular}
  \label{sfg:M}
  } \ 
  \subfigure[Sorted by $2$-order contexts]{\footnotesize
  \begin{tabular}{r|*{16}{p{1.75mm}}|}
    \cline{2-17}
    \textsf{5}  & $a$ & $a$ & $b$ & $a$ & $b$ & $c$ & $b$ & $c$ & $c$ & $b$ & $c$ & $b$ & $c$ & $a$ & $b$ & $b$ \\[-1mm]
    \textsf{2}  & $a$ & $b$ & $c$ & $b$ & $c$ & $c$ & $b$ & $c$ & $b$ & $c$ & $a$ & $b$ & $b$ & $a$ & $a$ & $b$ \\[-1mm]
    \textsf{4}  & $a$ & $b$ & $a$ & $b$ & $c$ & $b$ & $c$ & $c$ & $b$ & $c$ & $b$ & $c$ & $a$ & $b$ & $b$ & $a$ \\[-1mm]
    \textsf{8}  & $a$ & $b$ & $b$ & $a$ & $a$ & $b$ & $a$ & $b$ & $c$ & $b$ & $c$ & $c$ & $b$ & $c$ & $b$ & $c$ \\[-1mm]
    \textsf{3}  & $b$ & $a$ & $b$ & $c$ & $b$ & $c$ & $c$ & $b$ & $c$ & $b$ & $c$ & $a$ & $b$ & $b$ & $a$ & $a$ \\[-1mm]
    \textsf{6}  & $b$ & $a$ & $a$ & $b$ & $a$ & $b$ & $c$ & $b$ & $c$ & $c$ & $b$ & $c$ & $b$ & $c$ & $a$ & $b$ \\[-1mm]
    \textsf{7}  & $b$ & $b$ & $a$ & $a$ & $b$ & $a$ & $b$ & $c$ & $b$ & $c$ & $c$ & $b$ & $c$ & $b$ & $c$ & $a$ \\[-1mm]
    \textsf{1}  & $b$ & $c$ & $b$ & $c$ & $c$ & $b$ & $c$ & $b$ & $c$ & $a$ & $b$ & $b$ & $a$ & $a$ & $b$ & $a$ \\[-1mm]
    \textsf{10} & $b$ & $c$ & $a$ & $b$ & $b$ & $a$ & $a$ & $b$ & $a$ & $b$ & $c$ & $b$ & $c$ & $c$ & $b$ & $c$ \\[-1mm]
    \textsf{12} & $b$ & $c$ & $b$ & $c$ & $a$ & $b$ & $b$ & $a$ & $a$ & $b$ & $a$ & $b$ & $c$ & $b$ & $c$ & $c$ \\[-1mm]
    \textsf{15} & $b$ & $c$ & $c$ & $b$ & $c$ & $b$ & $c$ & $a$ & $b$ & $b$ & $a$ & $a$ & $b$ & $a$ & $b$ & $c$ \\[-1mm]
    \textsf{9}  & $c$ & $a$ & $b$ & $b$ & $a$ & $a$ & $b$ & $a$ & $b$ & $c$ & $b$ & $c$ & $c$ & $b$ & $c$ & $b$ \\[-1mm]
    \textsf{11} & $c$ & $b$ & $c$ & $a$ & $b$ & $b$ & $a$ & $a$ & $b$ & $a$ & $b$ & $c$ & $b$ & $c$ & $c$ & $b$ \\[-1mm]
    \textsf{13} & $c$ & $b$ & $c$ & $b$ & $c$ & $a$ & $b$ & $b$ & $a$ & $a$ & $b$ & $a$ & $b$ & $c$ & $b$ & $c$ \\[-1mm]
    \textsf{16} & $c$ & $b$ & $c$ & $c$ & $b$ & $c$ & $b$ & $c$ & $a$ & $b$ & $b$ & $a$ & $a$ & $b$ & $a$ & $b$ \\[-1mm]
    \textsf{14} & $c$ & $c$ & $b$ & $c$ & $b$ & $c$ & $a$ & $b$ & $b$ & $a$ & $a$ & $b$ & $a$ & $b$ & $c$ & $b$ \\
    \cline{2-17}
  \end{tabular}\label{sfg:M2}}
  \vspace*{-2mm}
  \caption{Computing the BWT and the ST of the word $w = bcbccbcbcabbaaba$}
  \label{fig:w:Mw:M2w}
\end{figure}

\section{The Burrows-Wheeler transform}
\label{sec:bwt}

The \emph{Burrows-Wheeler transform} (BWT) maps words $w$ of length
$n$ to pairs $(L,i)$ where $L$ is a word of length $n$ and $i$ is an
index in $\oneset{1, \ldots, n}$. The word $L$ is usually referred to
as the Burrows-Wheeler transform of $w$. In particular, the BWT is not
surjective. We will see below how the BWT works and that it is
one-to-one. It follows that only a fraction of $1/n$ of all possible
pairs $(L,i)$ appears as an image under the BWT. For instance
$(bacd,1)$ where $a<b<c<d$ is not an image under the BWT.

For $w \in \Sigma^+$ we define $M(w) = (w_1, \ldots, w_n)$ where
$\oneset{w_1, \ldots, w_n} = [w]$ and $w_1 \leq \cdots \leq w_n$.
Now, the Burrows-Wheeler transform of $w$ consists of the word
$\BWT(w) = \last(w_1) \cdots \last(w_n)$ and an index $i$ such that $w
= w_i$. Note that in contrast to the usual definition of the BWT, we
are using right shifts; at this point this makes no difference but it
unifies the presentation of succeeding transforms. At first glance, it
is surprising that one can reconstruct $M(w)$ from
$\BWT(w)$. Moreover, if we know the index $i$ of $w$ in the sorted
list $M(w)$, then we can reconstruct $w$ from $\BWT(w)$. One way of
how to reconstruct $M(w)$ is presented in the following lemma. For
later use, we prove a more general statement than needed for computing
the inverse of the BWT.

\begin{lemma}\label{lem:contextretrieval}
  Let $k \in \N$. Let $\bigcup_{i=1}^{s}\, [v_i] = \oneset{w_1,
    \ldots, w_n} \subseteq \Sigma^+$ be a multiset built from
  conjugacy classes $[v_i]$. Let $M = (w_1, \ldots, w_n)$ satisfy
  $\context_k(w_1) \leq \cdots \leq \context_k(w_n)$ and let $L =
  \last(w_1) \cdots \last(w_n)$ be the sequence of the last
  symbols. Then
  \begin{equation*}
    \context_k(w_i) = \lbl_{L}\pi_{L}(i) \cdot
    \lbl_{L}\pi_{L}^2(i)
    \,\cdots\, \lbl_{L}\pi_{L}^k(i)
  \end{equation*}
  where $\pi_{L}^t$ denotes the $t$-fold application of $\pi_{L}$ and
  $\lbl_L \pi_L^t (i) = \lbl_L\bigl(\pi_L^t(i)\bigr)$.
\end{lemma}

\begin{proof}
  By induction over the context length $t$, we prove that for all $i
  \in \oneset{1, \ldots, n}$ we have $\context_t(w_i) =
  \lbl_{L}\pi_{L}(i) \,\cdots\, \lbl_{L}\pi_{L}^t(i)$. For $t = 0$ we
  have $\context_0(w_i) = \varepsilon$ and hence, the claim is
  trivially true. Let now $0 < t \leq k$. By the induction hypothesis,
  the~$(t-1)$-order context of each $w_i$ is $\lbl_{L} \pi_{L}(i)
  \cdots \lbl_{L} \pi_{L}^{t-1}(i)$. By applying one right-shift, we
  see that the $t$-order context of $\rightshift(w_i)$ is $\lbl_{L}(i)
  \cdot \lbl_{L} \pi_{L}^1(i) \cdots \lbl_{L} \pi_{L}^{t-1}(i)$.

  The list $M$ meets the sort order induced by $k$-order contexts. In
  particular, $(w_1, \ldots, w_n)$ is sorted by $(t-1)$-order
  contexts. Let $(u_1, \ldots, u_n)$ be a stable sort by $t$-order
  contexts of the right-shifts $(\rightshift(w_1), \ldots,
  \rightshift(w_n))$. The construction of $(u_1, \ldots, u_n)$ only
  requires a sorting of the first letters of $(\rightshift(w_1),
  \ldots, \rightshift(w_n))$ such that identical letters preserve
  their order. The sequence of first letters of the words
  $\rightshift(w_1), \ldots, \rightshift(w_n)$ is exactly $L$. By
  construction of $\pi_L$, it follows that $(u_1, \ldots, u_n) =
  (w_{\pi_L(1)}, \ldots,
  w_{\pi_L(n)})$. Since $M$ is built from conjugacy
  classes, the multisets of elements occurring in $(w_1, \ldots, w_n)$
  and $(\rightshift(w_1), \ldots, \rightshift(w_n))$ are
  identical. The same holds for the multisets induced by $(w_1,
  \ldots, w_n)$ and $(u_1, \ldots, u_n)$. Therefore, the sequences of
  $t$-order contexts induced by $(w_1, \ldots w_n)$ and $(u_1, \ldots,
  u_n)$ are identical. Moreover, we conclude
  \begin{align*}
    \context_t(w_i) 
    = \context_t(u_i)
    = \context_t(w_{\pi_L(i)})
    = \lbl_{L}\pi_{L}(i) \cdot \lbl_L \pi_L^2(i) 
    \,\cdots\, \lbl_{L}\pi_{L}^t(i)
  \end{align*}
  which completes the induction. We note that in general $u_i \neq
  w_i$ since the sort order of $M$ beyond $k$-order contexts is
  arbitrary. Moreover, for $t=k+1$ the property $\context_t(w_i) =
  \context_t(u_i)$ does not need to hold (even though the multisets
  of $(k+1)$-order contexts coincide).
  \qed
\end{proof}

Note that in Lemma~\ref{lem:contextretrieval} we do not require that
all $v_i$ have the same length. Applying the BWT to conjugacy classes
of words with different lengths has also been used for the
\emph{Extended BWT} \cite{mrrs07tcs}. %mk modified

\begin{corollary}
  The BWT is invertible, i.e., given $(\BWT(w),i)$ where $i$ is the
  index of~$w$ in $M(w)$ one can reconstruct the word $w$.
\end{corollary}

\begin{proof}
  We set $k = \abs{w}$. Let $M = M(w)$ and $L = \BWT(w)$. Now, by
  Lemma~\ref{lem:contextretrieval} we see that
  \begin{equation*}
    w = w_i = \context_k(w_i) = \lbl_L\pi_L^1(i) \cdots \lbl_L\pi_L^{\abs{L}}(i).
  \end{equation*}
  In particular, $w = \lbl_L\pi_L^1(i) \cdots
  \lbl_L\pi_L^{\abs{L}}(i)$ only depends on $L$ and $i$.  \qed
\end{proof}

\begin{remark}
  In the special case of the BWT it is possible to compute the $i$-th
  element $w_i$ of $M(w)$ by using the inverse $\pi_L^{-1}$ of the
  permutation $\pi_L$: 
  \begin{equation*}
    w_i = \lbl_L\pi_L^{-\abs{w_i}+1}(i) \cdots \lbl_L\pi_L^{-1}(i) \lbl_L(i).
  \end{equation*}
  This justifies the usual way of computing the inverse of
  $(\BWT(w),i)$ from right to left (by using the restriction of
  $\pi_L^{-1}$ to the cycle containing the element $i$). The
  motivation is that the (required cycle of the) inverse $\pi_L^{-1}$
  seems to be easier to compute than the standard permutation $\pi_L$.
\end{remark}

\begin{example}\label{exa:bwt}
  We compute the BWT of $w = bcbccbcbcabbaaba$ from
  Example~\ref{exa:start}. The lexicographically sorted list $M(w)$
  can be found in Figure~\ref{sfg:M}.  This yields the transform
  $(\BWT(w),i) = (bacbbaaccacbbcbb,10)$ where $L = \BWT(w)$ is the
  last column of the matrix $M(w)$ and $w$ is the $i$-th row in
  $M(w)$. The standard permutation of $L$ is
  \begin{equation*}
    \pi_L = \left(\begin{array}{*{16}{p{5.5mm}}}
        1 & 2 & 3 & 4 & 5 & 6 & 7 & 8 & 9 & 10 & 11 & 12 & 13 & 14 & 15 & 16 \\
        2 & 6 & 7 & 10 & 1 & 4 & 5 & 12 & 13 & 15 & 16 & 3 & 8 & 9 & 11 & 14
      \end{array}\right).
  \end{equation*}
  Now, $\pi_L^1(10) \cdots \pi_L^{16}(10)$ gives us the  following
  sequence of positions starting with $\pi_L(10) = 15$:
  \begin{equation*}
    15 
    \stackrel{\pi_L}{\mapsto} 11
    \stackrel{\pi_L}{\mapsto} 16 
    \stackrel{\pi_L}{\mapsto} 14 
    \stackrel{\pi_L}{\mapsto} 9 
    \stackrel{\pi_L}{\mapsto} 13 
    \stackrel{\pi_L}{\mapsto} 8 
    \stackrel{\pi_L}{\mapsto} 12 
    \stackrel{\pi_L}{\mapsto} 3 
    \stackrel{\pi_L}{\mapsto} 7
    \stackrel{\pi_L}{\mapsto} 5 
    \stackrel{\pi_L}{\mapsto} 1 
    \stackrel{\pi_L}{\mapsto} 2 
    \stackrel{\pi_L}{\mapsto} 6 
    \stackrel{\pi_L}{\mapsto} 4 
    \stackrel{\pi_L}{\mapsto} 10.
  \end{equation*}
  Applying the labeling function $\lbl_L$ to this sequence of positions yields
  \begin{align*}
    & \ 
    \lbl_L(15)
    \lbl_L(11)
    \lbl_L(16)
    \lbl_L(14)
    \lbl_L(9)
    \lbl_L(13)
    \lbl_L(8)
    \lbl_L(12) \\
    & \quad
    \cdot \lbl_L(3)
    \lbl_L(7)
    \lbl_L(5)
    \lbl_L(1)
    \lbl_L(2)
    \lbl_L(6)
    \lbl_L(4)
    \lbl_L(10) \\
    &= bcbccbcbcabbaaba = w,
  \end{align*}
  i.e., we have successfully reconstructed the input $w$ from $(\BWT(w),i)$.
\end{example}

\section{The bijective Burrows-Wheeler transform}
\label{sec:lbwt}

Now we are ready to give a comprehensive description of Scott's
bijective variant of the BWT and to prove its correctness.  It maps a
word of length~$n$ to a word of length~$n$---without any index or
end-of-string symbol being involved. The key ingredient is the Lyndon
factorization: Suppose we are computing the BWT of a Lyndon word $v$,
then we do not need an index since we know that $v$ is the first
element of the list $M(v)$. This leads to the computation of a
multi-word BWT of the Lyndon factors of the input.

The bijective BWT of a word $w$ of length $n$ is defined as follows.
Let $w = v_s \cdots v_1$ with $v_s \geq \cdots \geq v_1$ be the Lyndon
factorization of $w$. Let $\bijM(w) = (u_1, \ldots, u_n)$ where $u_1
\leq^\omega \cdots \leq^{\omega} u_n$ and where the multiset
$\oneset{u_1, \ldots, u_n} = \bigcup_{i=1}^s [v_i]$.  Then, the
bijective BWT of $w$ is $\bijBWT(w) = \last(u_1) \cdots
\last(u_n)$. The \emph{S} in $\bijBWT$ is for \emph{Scottified}.
% Another common name for the bijective BWT is
% Burrows-Wheeler transform Scottified (BWTS). We are using $\bijBWT$ in
% order to unify presentation with that of the bijective sort transform.
% The letter \emph{L} in $\bijM$ and $\bijBWT$ is for
% \emph{Lyndon}.
Note that if $w$ is a power of a Lyndon word, then
$\bijBWT(w) = \BWT(w)$.

In some sense, the bijective BWT can be thought of as the composition
of the Lyndon factorization \cite{cfl58ann} with the inverse of the
Gessel-Reutenauer transform \cite{gr93jct}. In particular, a first
step towards a bijective BWT can be found in a 1993 article by Gessel
and Reutenauer \cite{gr93jct} (prior to the publication of the BWT
\cite{bw94tr}). The link between the Gessel-Reutenauer transform and
the BWT was pointed out later by Crochemore et al.\ \cite{cdp05tcs}.
%mk new
A similar approach as in the bijective BWT has been employed by
Mantaci et al.\ \cite{mrrs05cpm}; instead of the Lyndon factorization
they used a decomposition of the input into blocks of equal
length. The output of this variant is a word and a sequence of indices
(one for each block).
%mk end new
In its current form, the bijective BWT has been proposed by Scott
\cite{scott09} in a newsgroup posting in 2007.  Gil and Scott gave an
accessible version of the transform, an independent proof of its
correctness, and they
tested its performance in data compression \cite{gs09submitted}.  The
outcome of these tests is that the bijective BWT beats the usual BWT
on almost all files of the Calgary Corpus \cite{bwc89acm} by at least
a few hundred bytes which exceeds the gain of just saving the rotation
index.

\begin{lemma}\label{lem:cycleretrieval}
  Let $w = v_s \cdots v_1$ with $v_s \geq \cdots \geq v_1$ be the
  Lyndon factorization of~$w$, let $\bijM(w) = (u_1, \ldots, u_n)$,
  and let $L = \bijBWT(w)$. Consider the cycle $C$ of the permutation
  $\pi_L$ which contains the element $1$ and let $d$ be the length of
  $C$. Then $\lbl_L \pi_L^1(1) \cdots \lbl_L \pi_L^{d}(1) = v_1$.
\end{lemma}

\begin{proof}
  By Lemma~\ref{lem:contextretrieval} we see that $\bigl(
  \lbl_L \pi_L^1(1) \cdots \lbl_L \pi_L^{d}(1) \bigr)^{\abs{v_1}} =
  v_1^{d}$. Since $v_1$ is primitive it follows $\lbl_L
  \pi_L^1(1) \cdots \lbl_L \pi_L^{d}(1) = v_1^z$ for some $z \in
  \N$. In particular, the Lyndon factorization of $w$ ends with
  $v_1^z$.

  Let $U$ be the subsequence of $\bijM(w)$ which consists of those
  $u_i$ which come from this last factor $v_1^z$. The sequence $U$
  contains each right-shift of $v_1$ exactly $z$ times. Moreover, the
  sort-order within $U$ depends only on $\abs{v_1}$-order contexts.

  The element $v_1 = u_1$ is the first element in $U$ since $v_1$ is a
  Lyndon word. In particular, $\pi_L^0(1) = 1$ is the first occurrence
  of $\rightshift^{0}(v_1) = v_1$ within $U$. Suppose $\pi_L^j(1)$ is
  the first occurrence of $\rightshift^{j}(v_1)$ within $U$. Let
  $\pi_L^{j}(1) = i_1 < \cdots < i_z$ be the indices of all
  occurrences of $\rightshift^{j}(v_1)$ in $U$. By construction of
  $\pi_L$, we have $\pi_L(i_1) < \cdots < \pi_L(i_z)$ and therefore
  $\pi_L^{j+1}(1)$ is the first occurrence of $\rightshift^{j+1}(v_1)$
  within $U$.  Inductively, $\pi_L^j(1)$ always refers to the first
  occurrence of $\rightshift^{j}(v_1)$ within $U$ (for all $j \in
  \N$).  In particular it follows that $\pi_L^{\abs{v_1}}(1) = 1$ and
  $z=1$.
  \qed
\end{proof}

\begin{theorem}\label{thm:bijbwt}
  The bijective BWT is invertible, i.e., given $\bijBWT(w)$ one can
  reconstruct the word $w$.
\end{theorem}

\begin{proof}
  Let $L = \bijBWT(w)$ and let $w = v_s \cdots v_1$ with $v_s \geq
  \cdots \geq v_1$ be the Lyndon factorization of $w$. Each
  permutation admits a cycle structure.  We decompose the standard
  permutation $\pi_L$ into cycles $C_1, \ldots, C_t$. Let $i_j$ be the
  smallest element of the cycle $C_j$ and let $d_j$ be the length of
  $C_j$. We can assume that $1 = i_1 < \cdots < i_t$.

  We claim that $t = s$, $d_j = \abs{v_j}$, and $\lbl_L \pi_L^1(i_j)
  \cdots \lbl_L \pi_L^{d_j}(i_j) = v_j$.  By
  Lemma~\ref{lem:cycleretrieval} we have $\lbl_L \pi_L^1(i_1) \cdots
  \lbl_L \pi_L^{d_1}(i_1) = v_1$. Let $\pi_L'$ denote the restriction
  of $\pi_L$ to the set $C = C_2 \cup \cdots \cup C_t$, where by abuse
  of notation $C_2 \cup \cdots \cup C_t$ denotes the set of all
  elements occurring in $C_2, \ldots, C_t$.  Let $L' = \bijBWT(v_s
  \cdots v_2)$. The word $L'$ can be obtained from $L$ by removing all
  positions occurring in the cycle $C_1$. This yields a monotone
  bijection
  \begin{equation*}
    \alpha : C \to \oneset{1, \ldots, \abs{L'}}
  \end{equation*}
  such that $\lbl_L(i) = \lbl_{L'}\alpha(i)$ and $\alpha\pi_L(i) =
  \pi_{L'}\alpha(i)$ for all $i \in C$. In particular, $\pi_{L'}$ has
  the same cycle structure as $\pi_L'$ and $1 = \alpha(i_2) < \cdots <
  \alpha(i_t)$ is the sequence of the minimal elements within
  the cycles. By induction on the number of Lyndon factors,
  \begin{align*}
    v_s \cdots v_2 
    &= 
    \lbl_{L'} \pi_{L'}^1\alpha(i_t) \cdots \lbl_{L'} \pi_{L'}^{d_t} \alpha(i_t) 
    \;\cdots\;
    \lbl_{L'} \pi_{L'}^1\alpha(i_2) \cdots \lbl_{L'} \pi_{L'}^{d_2}(i_2) \\
    &=
    \lbl_{L'}\alpha\pi_{L}^1(i_t) \cdots \lbl_{L'}\alpha\pi_{L}^{d_t}(i_t) 
    \;\cdots\;
    \lbl_{L'}\alpha\pi_{L}^1(i_2) \cdots \lbl_{L'}\alpha\pi_{L}^{d_2}(i_2) \\
    &= 
    \lbl_{L}\pi_{L}^1(i_t) \cdots \lbl_{L}\pi_{L}^{d_t}(i_t) 
    \;\cdots\;
    \lbl_{L}\pi_{L}^1(i_2) \cdots \lbl_{L}\pi_{L}^{d_2}(i_2).
  \end{align*}
  Appending $\lbl_L \pi_L^1(i_1) \cdots \lbl_L \pi_L^{d_1}(i_1) =
  v_1$ to the last line allows us to reconstruct $w$ by
  \begin{equation*}
    w = \lbl_L \pi_L^1(i_t) \cdots \lbl_L \pi_L^{d_t}(i_t) 
    \;\cdots\;
    \lbl_L \pi_L^1(i_1) \cdots \lbl_L \pi_L^{d_1}(i_1).
  \end{equation*}
  Moreover, $t=s$ and $d_j = \abs{v_j}$. We note that this formula for
  $w$ only depends on $L$ and does not require any index to an element
  in $\bijM(w)$.
  \qed
\end{proof}

\begin{example}\label{exa:bijbwt}
  We again consider the word $w = bcbccbcbcabbaaba$ from
  Example~\ref{exa:start} and its Lyndon factorization $w = v_6 \cdots
  v_1$ where $v_6 = bcbcc$, $v_5 = bc$, $v_4 = bc$, $v_3 = abb$, $v_2
  = aab$, and $v_1 = a$. The lists $([v_1], \ldots, [v_6])$ and $\bijM(w)$
  are:
  \begin{center} \footnotesize
    \begin{tabular}{r|*{5}{p{1.75mm}}|}
      \cline{2-6}
      & \multicolumn{5}{|c|}{\ $([v_1], \ldots, [v_6])\ $\vphantom{$M^{M^m}$}} \\[0.1mm]
      \cline{2-6}
      \textsf{1}  & $a$ &   &   &   & \\[-1mm] %v_1
      \textsf{2}  & $a$ & $a$ & $b$ &   & \\[-1mm] %v_2
      \textsf{3}  & $b$ & $a$ & $a$ &   & \\[-1mm]
      \textsf{4}  & $a$ & $b$ & $a$ &   & \\[-1mm]
      \textsf{5}  & $a$ & $b$ & $b$ &   & \\[-1mm] %v_3
      \textsf{6}  & $b$ & $a$ & $b$ &   & \\[-1mm]
      \textsf{7}  & $b$ & $b$ & $a$ &   & \\[-1mm]
      \textsf{8}  & $b$ & $c$ &   &   & \\[-1mm] %v_4
      \textsf{9}  & $c$ & $b$ &   &   & \\[-1mm]
      \textsf{10} & $b$ & $c$ &   &   & \\[-1mm] %v_5
      \textsf{11} & $c$ & $b$ &   &   & \\[-1mm]
      \textsf{12} & $b$ & $c$ & $b$ & $c$ & $c$ \\[-1mm] %v_6
      \textsf{13} & $c$ & $b$ & $c$ & $b$ & $c$ \\[-1mm]
      \textsf{14} & $c$ & $c$ & $b$ & $c$ & $b$ \\[-1mm]
      \textsf{15} & $b$ & $c$ & $c$ & $b$ & $c$ \\[-1mm]
      \textsf{16} & $c$ & $b$ & $c$ & $c$ & $b$ \\
      \cline{2-6}
    \end{tabular}
    \qquad\qquad
    \begin{tabular}{r|*{5}{p{1.75mm}}|}
      \cline{2-6}
      & \multicolumn{5}{|c|}{\ \ \ \ \;\,$\bijM(w)$\vphantom{$M^{M^m}$}\,\;\ \ \ $\ $} \\[0.1mm]
      \cline{2-6}
      \textsf{1}  & $a$ &     &   &   & \\[-1mm] %v_1      
      \textsf{2}  & $a$ & $a$ & $b$ &   & \\[-1mm] %v_2
      \textsf{4}  & $a$ & $b$ & $a$ &   & \\[-1mm]
      \textsf{5}  & $a$ & $b$ & $b$ &   & \\[-1mm] %v_3
      \textsf{3}  & $b$ & $a$ & $a$ &   & \\[-1mm]
      \textsf{6}  & $b$ & $a$ & $b$ &   & \\[-1mm]
      \textsf{7}  & $b$ & $b$ & $a$ &   & \\[-1mm]
      \textsf{8}  & $b$ & $c$ &   &   & \\[-1mm] %v_4
      \textsf{10} & $b$ & $c$ &   &   & \\[-1mm] %v_5
      \textsf{12} & $b$ & $c$ & $b$ & $c$ & $c$ \\[-1mm] %v_6
      \textsf{15} & $b$ & $c$ & $c$ & $b$ & $c$ \\[-1mm]
      \textsf{9}  & $c$ & $b$ &   &   & \\[-1mm]
      \textsf{11} & $c$ & $b$ &   &   & \\[-1mm]
      \textsf{13} & $c$ & $b$ & $c$ & $b$ & $c$ \\[-1mm]
      \textsf{16} & $c$ & $b$ & $c$ & $c$ & $b$ \\[-1mm]
      \textsf{14} & $c$ & $c$ & $b$ & $c$ & $b$ \\
      \cline{2-6}
    \end{tabular}
  \end{center}
  Hence, we obtain $L = \bijBWT(w) = abababaccccbbcbb$ as the sequence
  of the last symbols of the words in $\bijM(w)$. The standard
  permutation $\pi_L$ induced by $L$ is
  \begin{equation*}
    \pi_L = \left(\begin{array}{*{16}{p{5.5mm}}}
        1 & 2 & 3 & 4 & 5 & 6 & 7 & 8 & 9 & 10 & 11 & 12 & 13 & 14 & 15 & 16 \\
        1 & 3 & 5 & 7 & 2 & 4 & 6 & 12 & 13 & 15 & 16 & 8 & 9 & 10 & 11 & 14
      \end{array}\right)
  \end{equation*}
  The cycles of $\pi_L$ arranged by their smallest elements are $C_1 =
  (1)$, $C_2 = (2, 3, 5)$, $C_3 = (4, 7, 6)$, $C_4 = (8, 12)$, $C_5 =
  (9, 13)$, and $C_6 = (10, 15, 11, 16, 14)$. Applying the labeling
  function $\lbl_L$ to the cycle $C_i$ (starting with the second
  element) yields the Lyndon factor $v_i$. With this procedure, we
  reconstructed $w = v_6 \cdots v_1$ from $L = \bijBWT(w)$.
\end{example}

\section{The sort transform}\label{sec:sorttransform}
\label{sec:st}

The \emph{sort transform} (ST) is a BWT
where we only sort the conjugates of the input up to a given depth $k$
and then we are using the index of the conjugates as a tie-breaker.
Depending on the depth $k$ and the implementation details this can
speed up compression (while at the same time slightly slowing down
decompression).

In contrast to the usual presentation of the ST, we are using right
shifts. This defines a slightly different version of the ST. The
effect is that the order of the symbols occurring in some particular
context is reversed. This makes sense, because in data compression the
ST is applied to the reversal of a word. Hence, in the ST of the
reversal of $w$ the order of the symbols in some particular context is
the same as in $w$. More formally, suppose $\overline{w} = x_0 c a_1
x_1 c a_2 x_2 \cdots c a_s x_s$ for $c \in \Sigma^+$ then in the sort
transform of order $\abs{c}$ of $w$, the order of the occurrences of
the letters $a_i$ is not changed.  This property can enable better
compression ratios on certain data.

While the standard permutation is induced by a sequence of letters
(i.e., a word) we now generalize this concept to sequences of words.
For a list of non-empty words $V = (v_1, \ldots, v_n)$ we now define
the \emph{$k$-order standard permutation} $\nu_{k,V}$ induced by $V$.
As for the standard permutation, the first step is the construction of
a new linear order $\preceq$ on $\oneset{1,\ldots,n}$. We define $i
\preceq j$ by the condition
\begin{equation*}
  \context_k(v_{i}) < \context_k(v_{j})
  \qquad \text{or} \qquad
  \context_k(v_{i}) = \context_k(v_{j}) \,\text{ and }\,
  i \leq j.
\end{equation*}
Let $j_1 \prec \cdots \prec j_n$ be the linearization of
$\oneset{1,\ldots,n}$ according to this new order. The idea is that we
sort the line numbers of $v_1, \ldots, v_n$ by first considering the
$k$-order contexts and, if these are equal, then use the line numbers
as tie-breaker.  As before, the linearization according to $\preceq$
induces a permutation $\nu_{k,V}$ by setting $\nu_{k,V}(i) = j_i$.
Now, $\nu_{k,V}(i)$ is the position of $v_i$ if we are sorting $V$ by
$k$-order context such that the line numbers serve as tie-breaker.
We set $M_k(v_1, \ldots, v_n) = (w_1, \ldots, w_n)$ where $w_i =
v_{\nu_{k,V}(i)}$. Now, we are ready to define the sort transform of
order $k$ of a word~$w$: Let $M_k([w]) = (w_1, \ldots, w_n)$; then
$\ST_k(w) = \last(w_1) \cdots \last(w_n)$, i.e., we first sort all
cyclic right-shifts of $w$ by their $k$-order contexts (by using a
stable sort method) and then we take the sequence of last symbols
according to this new sort order as the image under $\ST_k$. Since the
tie-breaker relies on right-shifts, we have $\ST_0(w) = \overline{w}$,
i.e., $\ST_0$ is the reversal mapping. The $k$-order sort transform of
$w$ is the pair $(\ST_k(w),i)$ where $i$ is the index of $w$ in
$M_k([w])$. As for the BWT, we see that the $k$-order sort transform
is not bijective.

Next, we show that it is possible to reconstruct $M_k([w])$ from
$\ST_k(w)$. Hence, it is possible to reconstruct $w$ from the pair
$(\ST_k(w),i)$ where $i$ is the index of~$w$ in~$M_k([w])$. The
presentation of the back transform is as follows. First, we will
introduce the \emph{$k$-order context graph} $G_k$ and we will show
that it is possible to rebuild~$M_k([w])$ from $G_k$. Then we will
show how to construct $G_k$ from $\ST_k(w)$. Again, the approach will
be slightly more general than required at the moment; but we will be
able to reuse it in the presentation of a bijective ST.

Let $V = ([u_1], \ldots, [u_s]) = (v_1, \ldots, v_n)$ be a list of
words built from conjugacy classes $[u_i]$ of non-empty words
$u_i$. Let $M = (w_1, \ldots, w_n)$ be an arbitrary permutation of the
elements in $V$.  We are now describing the edge-labeled directed
graph~\mbox{$G_k(M)$~--~the} $k$-order context graph of $M$~--~which will be
used later as a presentation tool for the inverses of the ST and the
bijective ST. The vertices of $G_k(M)$ consist of all $k$-order
contexts $\context_k(w)$ of words $w$ occurring in $M$. We draw an
edge $(c_1, i, c_2)$ from context $c_1$ to context $c_2$ labeled by
$i$ if $c_1 = \context_k(w_i)$ and $c_2 =
\context_k(\rightshift(w_i))$. Hence, every index $i \in \oneset{1,
  \ldots, n}$ of $M$ defines a unique edge in $G_k(M)$.  We can also
think of $\last(w_i)$ as an additional implicit label of the edge
$(c_1, i, c_2)$, since $c_2 = \context_k(\last(w_i) c_1)$.

A \emph{configuration} $(\mathcal{C},c)$ of the $k$-order context
graph $G_k(M)$ consists of a subset of the edges $\mathcal{C}$ and a
vertex $c$. The idea is that (starting at context $c$) we are walking
along the edges of $G_k(M)$ and whenever an edge is used, it is
removed from the set of edges $\mathcal{C}$. We now define the
transition
\begin{equation*}
  (\mathcal{C}_1,c_1)
  \stackrel{u}{\to} (\mathcal{C}_2,c_2)
\end{equation*}
from a configuration $(\mathcal{C}_1,c_1)$ to another configuration
$(\mathcal{C}_2,c_2)$ with output $u \in
\Sigma^*$ more formally. If there exists an edge in
$\mathcal{C}_1$ starting at $c_1$ and if $(c_1,i,c_2) \in
\mathcal{C}_1$ is the unique edge with the smallest label $i$ starting
at $c_1$, then we have the single-step transition
\begin{equation*}
  (\mathcal{C}_1,c_1)
  \stackrel{a}{\to} (\mathcal{C}_1 \setminus
  \smallset{(c_1,i,c_2)},c_2) \qquad \text{where } a = \last(w_i)
\end{equation*}
If there is no edge in $\mathcal{C}_1$ starting at $c_1$, then the
outcome of \,$(\mathcal{C}_1,c_1) \stackrel{}{\to}$\, is undefined.
Inductively, we define $(\mathcal{C}_1,c_1)
\stackrel{\varepsilon}{\to} (\mathcal{C}_1,c_1)$ and for $a \in \Sigma$ and
$u \in \Sigma^*$ we have
\begin{equation*}
  (\mathcal{C}_1,c_1)
  \stackrel{au}{\to} (\mathcal{C}_2,c_2)
  \quad \text{if} \quad
  (\mathcal{C}_1,c_1) \stackrel{u}{\to} (\mathcal{C}',c')
  \,\text{ and }\,
  (\mathcal{C}',c') \stackrel{a}{\to} (\mathcal{C}_2,c_2)
\end{equation*}
for some configuration $(\mathcal{C}',c')$. Hence, the reversal
$\overline{au}$ is the label along the path of length $\abs{au}$
starting at configuration $(\mathcal{C}_1,c_1)$. In particular, if
$(\mathcal{C}_1,c_1) \stackrel{u}{\to} (\mathcal{C}_2,c_2)$ holds,
then it is possible to chase at least $\abs{u}$ transitions starting
at $(\mathcal{C}_1,c_1)$; vice versa, if we are chasing $\ell$
transitions then we obtain a word of length $\ell$ as a label.  We
note that successively taking the edge with the smallest label comes
from the use of right-shifts.  If we had used left-shifts we would
have needed to chase largest edges for the following lemma to
hold. The reverse labeling of the big-step transitions is motivated by
the reconstruction procedure which will work from right to left.

\begin{lemma}\label{lem:smallestedgechasing}
  Let $k \in \N$, $V = ([v_1], \ldots, [v_s])$, $c_i =
  \context_k(v_i)$, and $G = G_k(M_k(V))$. Let $\mathcal{C}_1$
  consist of all edges of $G$. Then
  \begin{align*}
    (\mathcal{C}_1,c_1) &\stackrel{v_1}{\to} (\mathcal{C}_2,c_1) \\
    (\mathcal{C}_2,c_2) &\stackrel{v_2}{\to} (\mathcal{C}_3,c_2) \\[-2mm]
    &\ \ \vdots \\
    (\mathcal{C}_s,c_s) &\stackrel{v_s}{\to} (\mathcal{C}_{s+1},c_s) .
  \end{align*}
\end{lemma}

\begin{proof}
  Let $M_k(V) = (w_1, \ldots, w_n)$. Consider some index $i$, $1 \leq
  i \leq s$, and let $(u_1, \ldots, u_t) = ([v_1], \ldots,
  [v_{i-1}])$.  Suppose that $\mathcal{C}_i$ consists of all edges of
  $G$ except for those with labels $\nu_{k,V}(j)$ for $1 \leq j \leq
  t$. Let $q = \abs{v_i}$. We write $v_i = a_1 \cdots a_q$ and
  $u_{t+j} = \rightshift^{j-1}(v_i)$, i.e., $[v_i] = (u_{t+1}, \ldots,
  u_{t+q})$.  Starting with
  $(\mathcal{C}_{i,1},c_{i,1})=(\mathcal{C}_i,c_i)$, we show that
  the sequence of transitions
  \begin{equation*}
    (\mathcal{C}_{i,1},c_{i,1}) \stackrel{a_q}{\to}
    (\mathcal{C}_{i,2},c_{i,2}) \stackrel{a_{q-1}}{\to}
    \cdots\,
    (\mathcal{C}_{i,q},c_{i,q}) \stackrel{a_1}{\to}
    (\mathcal{C}_{i,q+1},c_{i,q+1})
  \end{equation*}
  is defined. More precisely, we will see that the transition
  $(\mathcal{C}_{i,j},c_{i,j}) \stackrel{a_{q+1-j}}{\longrightarrow}
  (\mathcal{C}_{i,j+1},c_{i,j+1})$ walks along the edge
  $\bigl(c_{i,j}, \nu_{k,V}(t+j), c_{i,j+1}\bigr)$ and hence indeed is
  labeled with the letter $a_{q+1-j} = \last(u_{t+j}) =
  \last(w_{\nu_{k,V}(t+j)})$. Consider the context $c_{i,j}$.  By
  induction, we have $c_{i,j} = \context_k(u_{t+j})$ and no edge with
  label $\nu_{k,V}(\ell)$ for $1 \leq \ell < t+j$ occurs
  in~$\mathcal{C}_{i,j}$ while all other labels do occur. In
  particular, $(c_{i,j}, \nu_{k,V}(t+j), c_{i,j+1})$ for $c_{i,j+1} =
  \context_k(\rightshift(u_{t+j})) = \context_k(u_{t+j+1})$ is an edge
  in $\mathcal{C}_{i,j}$ (where $\context_k(\rightshift(u_{t+j})) =
  \context_k(u_{t+j+1})$ only holds for $j<q$; we will consider the
  case $j=q$ below). Suppose there were an edge $(c_{i,j}, z ,c') \in
  \mathcal{C}_{i,j}$ with $z < \nu_{k,V}(t+j)$.  Then $\context_k(w_z)
  = c_{i,j}$ and hence, $w_z$ has the same $k$-order context as
  $w_{\nu_{k,V}(t+j)}$. But in this case, in the construction of
  $M_k(V)$ we used the index in $V$ as a tie-breaker. It follows
  $\nu_{k,V}^{-1}(z) < t+1$ which contradicts the properties of
  $\mathcal{C}_{i,j}$. Hence, $(c_{i,j}, \nu_{k,V}(t+j), c_{i,j+1})$
  is the edge with the smallest label starting at context
  $c_{i,j}$. Therefore, $\mathcal{C}_{i,j+1} = \mathcal{C}_{i,j}
  \setminus \smallset{(c_{i,j}, \nu_{k,V}(t+j), c_{i,j+1})}$ and
  $(\mathcal{C}_{i,j},c_{i,j}) \stackrel{a_{q+1-j}}{\longrightarrow}
  (\mathcal{C}_{i,j+1},c_{i,j+1})$ indeed walks along the edge
  $(c_{i,j}, \nu_{k,V}(t+j), c_{i,j+1})$.

  It remains to verify that $c_{i,1} = c_{i,q+1}$, but this is clear
  since $c_{i,1} = \context_k(u_{t+1}) =
  \context_k(\rightshift^{q}(u_{t+1})) = c_{i,q+1}$.
  \qed
\end{proof}

\begin{lemma}\label{lem:reconstructingG}
  Let $k \in \N$, $V = ([v_1], \ldots, [v_s])$, $M = M_k(V) = (w_1,
  \ldots, w_n)$, and $L = \last(w_1) \cdots \last(w_n)$. Then it is
  possible to reconstruct $G_k(M)$ from $L$.
\end{lemma}

\begin{proof}
  By Lemma~\ref{lem:contextretrieval} it is possible to reconstruct
  the contexts $c_i = \context_k(w_i)$. This gives the vertices of the
  graph $G_k(M)$. Write $L = a_1 \cdots a_n$. For each $i \in
  \oneset{1, \ldots, n}$ we draw an edge $(c_i, i, \context_k(a_i
  c_i))$. This yields the edges of $G_k(M)$.
  \qed
\end{proof}

\begin{corollary}
  The $k$-order ST is invertible, i.e., given $(\ST_k(w),i)$ where $i$
  is the index of $w$ in $M_k([w])$ one can reconstruct the word $w$.
\end{corollary}

\begin{proof}
  The construction of $w$ consists of two phases. First, by
  Lemma~\ref{lem:reconstructingG} we can compute $G_k(M_k([w]))$.  By
  Lemma~\ref{lem:contextretrieval} we can compute $c = \context_k(w)$
  from $(\ST_k(w),i)$. In the second stage, we are using
  Lemma~\ref{lem:smallestedgechasing} for reconstructing $w$ by
  chasing 
  \begin{equation*}
    (\mathcal{C},c) \stackrel{w}{\to}
    (\emptyset,c)
  \end{equation*}
  where $\mathcal{C}$ consists of all edges in $G_k(M_k([w]))$.
  \qed
\end{proof}

Efficient implementations of the inverse transform rely on the fact
that the $k$-order contexts of $M_k([w])$ are ordered. This allows the
implementation of the $k$-order context graph $G_k$ in a vectorized
form \cite{abm08book,nz06dcc,nz07ieee,nzc08cpm}.

\begin{example}\label{exa:st}
  We compute the sort transform of order $2$ of $w = bcbccbcbcabbaaba$
  from Example~\ref{exa:start}. The list $M_2([w])$ is depicted in
  Figure~\ref{sfg:M2}.  This yields the transform $(\ST_2(w),i) =
  (bbacabaacccbbcbb,8)$ where $L = \ST_2(w)$ is the last column of the
  matrix~$M_2([w])$ and $w$ is the $i$-th element in $M_2([w])$. Next,
  we show how to reconstruct the input $w$ from $(L,i)$.  The standard
  permutation induced by $L$ is
  \begin{equation*}
    \pi_L = \left(\begin{array}{*{16}{p{5.5mm}}}
        1 & 2 & 3 & 4 & 5 & 6 & 7 & 8 & 9 & 10 & 11 & 12 & 13 & 14 & 15 & 16 \\
        3 & 5 & 7 & 8 & 1 & 2 & 6 & 12 & 13 & 15 & 16 & 4 & 9 & 10 & 11 & 14
      \end{array}\right).
  \end{equation*}
  Note that $\pi_L$ has four cycles $C_1 = (1,3,7,6,2,5)$, $C_2 =
  (4,8,12)$, $C_3 = (9,13)$, and $C_4 = (10,15,11,16,14)$.  We obtain
  the context of order $2$ of the $j$-th word by $c_j = \lbl_L\pi_L(j)
  \lbl_L\pi_L^2(j)$. In particular, $c_1 = aa$, $c_2 = c_3 = c_4 =
  ab$, $c_5 = c_6 = ba$, $c_7 = bb$, $c_8 = c_9 = c_{10} = c_{11} =
  bc$, $c_{12} = ca$, $c_{13} = c_{14} = c_{15} = cb$, and $c_{16} =
  cc$. With $L$ and these contexts we can construct the graph $G =
  G_2(M_2([w])$. The vertices of $G$ are the contexts and the
  edge-labels represent positions in $L$. The graph $G$ is depicted
  below:
  \begin{center} \footnotesize
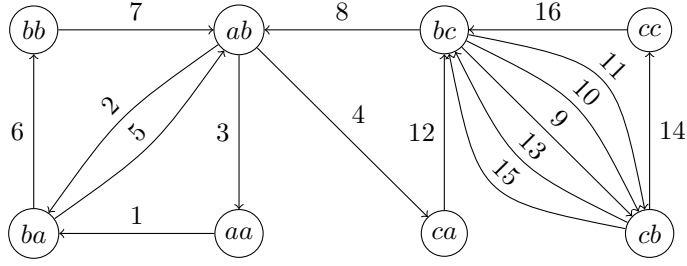
    \begin{tikzpicture}[scale=0.9]
      \draw (0,0) node[circle,draw,inner sep=2.5pt] (ba) {$ba$};
      \draw (3,0) node[circle,draw,inner sep=2.5pt] (aa) {$aa$};
      \draw (6,0) node[circle,draw,inner sep=2.5pt] (ca) {$ca$};
      \draw (9,0) node[circle,draw,inner sep=2.5pt] (cb) {$cb$};
      \draw (0,3) node[circle,draw,inner sep=2.5pt] (bb) {$bb$};
      \draw (3,3) node[circle,draw,inner sep=2.5pt] (ab) {$ab$};
      \draw (6,3) node[circle,draw,inner sep=2.5pt] (bc) {$bc$};
      \draw (9,3) node[circle,draw,inner sep=2.5pt] (cc) {$cc$};

      \draw[->] (aa) -- node[above] {$1$} (ba);
      \draw[->] (ab) .. controls (1.25,1.75) .. node[sloped,above] {$2$} (ba);
      \draw[->] (ab) -- node[left] {$3$} (aa);
      \draw[->] (ab) -- node[above right] {$4$} (ca);
      \draw[->] (ba) .. controls (1.75,1.25) .. node[sloped,above] {$5$} (ab);
      \draw[->] (ba) -- node[left] {$6$} (bb);
      \draw[->] (bb) -- node[above] {$7$} (ab);
      \draw[->] (bc) -- node[above] {$8$} (ab);
      \draw[->] (bc) -- node[sloped,above] {$9$} (cb);
      \draw[->] (bc) .. controls (8,2) .. node[sloped,above] {$10$} (cb);
      \draw[->] (bc) .. controls (8.5,2.5) .. node[sloped,above] {$11$} (cb);
      \draw[->] (ca) -- node[left] {$12$} (bc);
      \draw[->] (cb) .. controls (7,1) .. node[sloped,above] {$13$} (bc);
      \draw[->] (cb) -- node[right] {$14$} (cc);
      \draw[->] (cb) .. controls (6.5,0.5) .. node[sloped,above] {$15$} (bc);
      \draw[->] (cc) -- node[above] {$16$} (bc);
    \end{tikzpicture}
  \end{center}
  We are starting at the context $c_i = c_8 = bc$ and then we are
  traversing $G$ along the smallest edge-label amongst the unused
  edges. The sequence of the edge labels obtained this way is
  \begin{equation*}
    (8,2,5,3,1,6,7,4,12,9,13,10,14,16,11,15).
  \end{equation*}
  The labeling of this sequence of positions yields $\overline{w} =
  abaabbacbcbccbcb$.  Since we are constructing the input from right
  to left, we obtain $w = bcbccbcbcabbaaba$.
\end{example}

\section{The bijective sort transform}
\label{sec:lst}

The bijective sort transform combines the Lyndon factorization with
the ST. This yields a new algorithm which serves as a similar
preprocessing step in data compression as the BWT. In a lot of
applications, it can be used as a substitute for the ST. The proof of
the bijectivity of the transform is slightly more technical than the
analogous result for the bijective BWT. The main reason is that the
bijective sort transform is less modular than the bijective BWT (which
can be grouped into a `Lyndon factorization part' and a
`Gessel-Reutenauer transform part' and which for example allows the
use of different orders on the alphabet for the different parts).

For the description of the bijective ST and of its inverse, we rely on
notions from Section~\ref{sec:sorttransform}. The bijective ST of a
word $w$ of length $n$ is defined as follows.  Let $w = v_s \cdots
v_1$ with $v_s \geq \cdots \geq v_1$ be the Lyndon factorization of
$w$. Let $M_k([v_1], \ldots, [v_s]) = (u_1, \ldots, u_n)$. Then the
bijective ST of order $k$ of $w$ is $\bijST_k(w) = \last(u_1) \cdots
\last(u_n)$.  That is, we are sorting the conjugacy classes of the
Lyndon factors by $k$-order contexts and then take the sequence of the
last letters. The letter \emph{L} in $\bijST_k$ is for \emph{Lyndon}.

\begin{theorem}\label{thm:bijst}
  The bijective ST of order $k$ is invertible, i.e., given
  $\bijST_k(w)$ one can reconstruct the word $w$.
\end{theorem}

\begin{proof}
  Let $w = v_s \cdots v_1$ with $v_s \geq \cdots \geq v_1$ be the
  Lyndon factorization of $w$, let $c_i = \context_k(v_i)$, and let $L
  = \bijST_k(w)$. By Lemma~\ref{lem:reconstructingG} we can rebuild
  the $k$-order context graph $G = G_k(M_k([v_1], \ldots, [v_s])) =
  (w_1, \ldots, w_n)$ from $L$. Let $\mathcal{C}_1$ consist of all
  edges in $G$. Then by Lemma~\ref{lem:smallestedgechasing} we see
  that
  \begin{align*}
    (\mathcal{C}_1,c_1) &\stackrel{v_1}{\to} (\mathcal{C}_2,c_1) \\[-2mm]
%    (\mathcal{C}_2,c_2) &\stackrel{v_2}{\to} (\mathcal{C}_3,c_2) \\
    &\ \ \vdots \\
    (\mathcal{C}_s,c_s) &\stackrel{v_s}{\to} (\mathcal{C}_{s+1},c_s).
  \end{align*}
  We cannot use this directly for the reconstruction of $w$ since we
  do not know the Lyndon factors $v_i$ and the contexts $c_i$.

  The word $v_1$ is the first element in the list $M_k([v_1], \ldots,
  [v_s])$ because $v_1$ is lexicographically minimal and it appears as
  the first element in the list $([v_1], \ldots, [v_s])$. Therefore,
  by Lemma~\ref{lem:contextretrieval} we obtain $c_1 = \context_k(v_1)
  = \lbl_L\pi_L(1) \cdots \lbl_L\pi_L^k(1)$.

  The reconstruction procedure works from right to left.  Suppose we
  have already reconstructed $w' v_j \cdots v_1$ for $j \geq 0$ with
  $w'$ being a (possibly empty) suffix of $v_{j+1}$.  Moreover,
  suppose we have used the correct contexts $c_1, \ldots, c_{j+1}$.
  Consider the configuration $(\mathcal{C}',c')$ defined by
  \begin{align*}
    (\mathcal{C}_1,c_1) &\stackrel{v_1}{\to} (\mathcal{C}_2,c_1) \\[-2mm]
    &\ \ \vdots \\
    (\mathcal{C}_j,c_j) &\stackrel{v_j}{\to} (\mathcal{C}_{j+1},c_j) \\
    (\mathcal{C}_{j+1},c_{j+1}) &\stackrel{w'}{\to} (\mathcal{C}',c')
  \end{align*}
  We assume that the following invariant holds: $\mathcal{C}_{j+1}$
  contains no edges $(c'',\ell,c''')$ with $c'' < c_{j+1}$.
  We want to rebuild the next letter. We have to consider three cases.
  First, if $\abs{w'} < \abs{v_{j+1}}$ then
  \begin{equation*}
    (\mathcal{C}',c') \stackrel{a}{\to} (\mathcal{C}'',c'')
  \end{equation*}
  yields the next letter $a$ such that $a w'$ is a suffix of
  $v_{j+1}$. Second, let $\abs{w'} = \abs{v_{j+1}}$ and suppose that
  there exists an edge $(c_{j+1},\ell,c''') \in \mathcal{C}'$ starting
  at $c' = c_{j+1}$. Then there exists a word $v'$ in $[v_{j+2}],
  \ldots, [v_s]$ such that $\context_k(v') = c_{j+1}$. If
  $\context_k(v_{j+2}) \neq c_{j+1}$ then from the invariant it
  follows that $\context_k(v_{j+2}) > c_{j+1} = \context_k(v')$.
  This is a contradiction, since $v_{j+2}$ is minimal among the words
  in $[v_{j+2}], \ldots, [v_s]$. Hence, $\context_k(v_{j+2}) = c_{j+2} =
  c_{j+1}$ and the invariant still holds for $\mathcal{C}_{j+2} =
  \mathcal{C}'$. The last letter $a$ of $v_{j+2}$ is obtained by
  \begin{equation*}
    (\mathcal{C}',c') 
    = (\mathcal{C}_{j+2},c_{j+2}) \,\stackrel{a}{\to}\, (\mathcal{C}'',c'').
  \end{equation*}
  The third case is $\abs{w'} = \abs{v_{j+1}}$ and there is no edge
  $(c_{j+1},\ell,c''') \in \mathcal{C}'$ starting at $c' = c_{j+1}$.
  As before, $v_{j+2}$ is minimal among the (remaining) words in
  $[v_{j+2}], \ldots, [v_s]$. By construction of $G$, the unique edge
  $(c'',\ell,c''') \in \mathcal{C}'$ with the minimal label $\ell$ has
  the property that $w_{\ell} = v_{j+2}$. In particular, $c'' =
  c_{j+2}$. Since $v_{j+2}$ is minimal, the invariant for
  $\mathcal{C}_{j+2} = \mathcal{C}'$ is established.
  In this case, the last letter $a$ of $v_{j+2}$ is obtained by
  \begin{equation*}
    (\mathcal{C}_{j+2},c_{j+2}) \stackrel{a}{\to} (\mathcal{C}'',c''').
  \end{equation*}
  We note that we cannot distinguish between the first and the second
  case since we do not know the length of $v_{j+1}$, but in both
  cases, the computation of the next symbol is identical.  In
  particular, in contrast to the bijective BWT we do not implicitly
  recover the Lyndon factorization of $w$.
  \qed
\end{proof}

We note that the proof of Theorem~\ref{thm:bijst} heavily relies on
two design criteria. The first one is to consider $M_k([v_1], \ldots,
[v_s])$ rather than $M_k([v_s], \ldots, [v_1])$, and the second is to
use right-shifts rather than left-shifts.
The proof of Theorem~\ref{thm:bijst} yields the following algorithm
for reconstructing $w$ from $L = \bijST_k(w)$:
\begin{enumerate}[(1)]
\item Compute the $k$-order context graph $G = G_k$ and the $k$-order
  context $c_1$ of the last Lyndon factor of $w$.
\item Start with the configuration $(\mathcal{C},c)$ where
  $\mathcal{C}$ contains all edges of $G$ and $c \coloneq c_1$.
\item If there exists an outgoing edge starting at $c$ in the set
  $\mathcal{C}$, then
  \begin{itemize}
  \item Let $(c,\ell,c')$ be the edge with the minimal label $\ell$
    starting at $c$.
  \item Output $\lbl_L(\ell)$. 
  \item Set $\mathcal{C} \coloneq \mathcal{C} \setminus
    \smallset{(c,\ell,c')}$ and $c \coloneq c'$. 
  \item Continue with step (3).
  \end{itemize}
\item If there is no outgoing edge starting at $c$ in the set
  $\mathcal{C}$, but $\mathcal{C} \neq \emptyset$, then
  \begin{itemize}
  \item Let $(c',\ell,c'') \in \mathcal{C}$ be the edge with
    the minimal label $\ell$.
  \item Output $\lbl_L(\ell)$.
  \item Set $\mathcal{C} \coloneq \mathcal{C} \setminus
    \smallset{(c',\ell,c'')}$ and $c \coloneq c''$. 
  \item Continue with step (3).
  \end{itemize}
\item The algorithm terminates as soon as $\mathcal{C} = \emptyset$.
\end{enumerate}
The sequence of the outputs is the reversal $\overline{w}$ of the word
$w$.

\begin{example}
  We consider the word $w = bcbccbcbcabbaaba$ from
  Example~\ref{exa:start} and its Lyndon factorization $w = v_6 \cdots
  v_1$ where $v_6 = bcbcc$, $v_5 = bc$, $v_4 = bc$, $v_3 = abb$, $v_2
  = aab$, and $v_1 = a$.  For this particular word $w$ the bijective
  Burrows-Wheeler transform and the bijective sort transform of order
  $2$ coincide.  From Example~\ref{exa:bijbwt}, we know $L =
  \bijST_2(w) = \bijBWT(w) = abababaccccbbcbb$ and the standard
  permutation $\pi_L$.  As in Example~\ref{exa:st} we can reconstruct
  the $2$-order contexts $c_1, \ldots, c_{16}$ of
  $M_2([v_1],\ldots,[v_6])$: $c_1 = c_2 = aa$, $c_3 = c_4 = ab$, $c_5
  = c_6 = ba$, $c_7 = bb$, $c_8 = c_9 = c_{10} = c_{11} = bc$, $c_{12}
  = c_{13} = c_{14} = c_{15} = cb$, and $c_{16} = cc$. With $L$ and
  the $2$-order contexts we can construct the graph $G =
  G_k(M_2([v_1],\ldots,[v_6]))$:
  \begin{center} \footnotesize
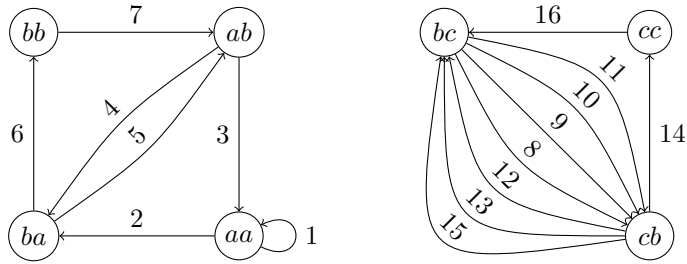
    \begin{tikzpicture}[scale=0.9]
      \draw (0,0) node[circle,draw,inner sep=2.5pt] (ba) {$ba$};
      \draw (3,0) node[circle,draw,inner sep=2.5pt] (aa) {$aa$};
      \draw (9,0) node[circle,draw,inner sep=2.5pt] (cb) {$cb$};
      \draw (0,3) node[circle,draw,inner sep=2.5pt] (bb) {$bb$};
      \draw (3,3) node[circle,draw,inner sep=2.5pt] (ab) {$ab$};
      \draw (6,3) node[circle,draw,inner sep=2.5pt] (bc) {$bc$};
      \draw (9,3) node[circle,draw,inner sep=2.5pt] (cc) {$cc$};

      \draw[->] (aa) .. controls (4,-0.5) and (4,+0.5) .. node[right] {$1$} (aa);
      \draw[->] (aa) -- node[above] {$2$} (ba);
      \draw[->] (ab) -- node[left] {$3$} (aa);
      \draw[->] (ab) .. controls (1.25,1.75) .. node[sloped,above] {$4$} (ba);
      \draw[->] (ba) .. controls (1.75,1.25) .. node[sloped,above] {$5$} (ab);
      \draw[->] (ba) -- node[left] {$6$} (bb);
      \draw[->] (bb) -- node[above] {$7$} (ab);
      \draw[->] (bc) .. controls (7,1) .. node[sloped,above] {$8$} (cb);
      \draw[->] (bc) -- node[sloped,above] {$9$} (cb);
      \draw[->] (bc) .. controls (8,2) .. node[sloped,above] {$10$} (cb);
      \draw[->] (bc) .. controls (8.5,2.5) .. node[sloped,above] {$11$} (cb);
      \draw[->] (cb) .. controls (6.5,0.5) .. node[sloped,above] {$12$} (bc);
      \draw[->] (cb) .. controls (6.0,0.0) .. node[sloped,above] {$13$} (bc);
      \draw[->] (cb) -- node[right] {$14$} (cc);
      \draw[->] (cb) .. controls (5.5,-0.5) .. node[sloped,above] {$15$} (bc);
      \draw[->] (cc) -- node[above] {$16$} (bc);
    \end{tikzpicture}
  \end{center}
  We are starting with the edge with label $1$ and then we are
  traversing $G$ along the smallest unused edges. If we end in a
  context with no outgoing unused edges, then we are continuing with
  the smallest unused edge. This gives the sequence $(1,2,5,3)$ after
  which we end in context $aa$ with no unused edges available. Then we
  continue with the sequences $(4,6,7)$ and
  $(8,12,9,13,10,14,16,11,15)$. The complete sequence of edge labels
  obtained this way is
  \begin{equation*}
    (1,2,5,3,\ 4,6,7,\ 8,12,9,13,10,14,16,11,15)
  \end{equation*}
  and the labeling of this sequence with $\lbl_L$ yields $\overline{w}
  = abaabbacbcbccbcb$. As for the ST, we are reconstructing the input
  from right to left, and hence we get $w = bcbccbcbcabbaaba$.
\end{example}

\section{Summary}

We discussed two bijective variants of the Burrows-Wheeler transform
(BWT).  The first one is due to Scott. Roughly speaking, it is a
combination of the Lyndon factorization and the
Gessel-Reuternauer transform.  The second variant is derived from the
sort transform (ST); it is the main contribution of this paper.
We gave full constructive proofs for the bijectivity of both
transforms.  As a by-product, we provided algorithms for the inverse
of the BWT and the inverse of the ST.  For the latter, we introduced
an auxiliary graph structure---the $k$-order context graph. This
graph yields an intermediate step in the computation of the inverse of
the ST and the bijective ST. It can be seen as a generalization of the
cycle decomposition of the standard permutation---which in turn can
be used as an intermediate step in the computation of the inverse of
the BWT and the bijective BWT.

\bigskip

\noindent
\textbf{Acknowledgments.} \ 
The author would like to thank Yossi Gil and David A.\ Scott for many
helpful discussions on this topic as well as Alexander Lauser,
%mk new (plus comma after "Lauser")
 Antonio Restivo,
%mk end new 
 and the
anonymous referees for their numerous suggestions which improved the
presentation of this paper.

%mk bibliography changed

%\bibliographystyle{psc}
%\bibliography{bwt}

\end{document}